\documentclass[letterpaper,11pt]{article}  
\usepackage{amsfonts}
\usepackage[pdftex]{graphicx}
\usepackage{amsmath,amssymb,amsthm}
\usepackage{natbib}
\usepackage{hyperref}
\usepackage{bm}
\usepackage{caption}
\usepackage{subcaption}
\usepackage{adjustbox,lipsum}
\usepackage{fancyhdr}
\usepackage{sectsty}
\usepackage{stmaryrd}
\usepackage{setspace}                      
\usepackage{booktabs}    
\usepackage{pdfsync}        
\usepackage[backgroundcolor=blue!40,linecolor=blue!40]{todonotes}
\usepackage{fancybox}
\usepackage{appendix}
\usepackage{pgfplots}
\usepackage{tikz}
\usepgfplotslibrary{fillbetween}

\usetikzlibrary{arrows,calc,patterns,math}
\usetikzlibrary{snakes}
\tikzset{
>=stealth',
help lines/.style={dashed, thick},
axis/.style={<->},
important line/.style={thick},
connection/.style={thick, dotted},
}

\definecolor{cornellred}{RGB}{179,27,27} 
\definecolor{cornellblue}{RGB}{00,00,170}
\definecolor{cornellgrey}{RGB}{96,94,92}

\newtheoremstyle{myplain}
  {9pt}
  {9pt}
  {\itshape}
  {\parindent}
  {\scshape}
  {:}
  {.5em}
  {}
\newtheoremstyle{mydefinition}
  {9pt}
  {9pt}
  {\itshape}
  {\parindent}
  {\scshape}
  {:}
  {.5em}
  {}
\newtheoremstyle{myremark}
  {9pt}
  {9pt}
  {}
  {\parindent}
  {\scshape}
  {:}
  {.5em}
  {}
\theoremstyle{myplain}
\newtheorem{theorem}{Theorem}[section]

\newtheorem{lemma}{Lemma}[section]

\theoremstyle{mydefinition}

\newtheorem{assumption}{Assumption}

\newtheorem{innercustomassumption}{Assumption}
\newenvironment{customassumption}[1]
  {\renewcommand\theinnercustomassumption{#1}\innercustomassumption}
  {\endinnercustomassumption}

\theoremstyle{myremark}

\newtheorem{remark}{Remark}[section]

\renewcommand{\cite}{\citet}

\newcommand{\R}{\mathbb{R}}

\newcommand{\cJ}{\mathcal{J}}
  
\newcommand{\cL}{\mathcal{L}}

\newcommand{\cT}{\mathcal{T}}

\usetikzlibrary{calc}
\def\centerarc[#1](#2)(#3:#4:#5){ \draw[#1] ($(#2)+({#5*cos(#3)},{#5*sin(#3)})$) arc (#3:#4:#5);}

\hypersetup{colorlinks=true, linkcolor=blue, citecolor=blue}                          
\numberwithin{equation}{section}

\begin{document}

\title{Constraint Qualifications in Partial Identification\thanks{We are grateful to Ivan Canay for conversations that helped bring this paper into focus and to three anonymous referees as well as the editor and co-editor for extremely careful readings of the manuscript. We also thank Isaiah Andrews, Lixiong Li, and seminar audiences at the Bristol Econometrics Study Group, Bristol/Warwick joint seminar, Columbia, and Duke for their feedback. Any and all errors are our own. We gratefully acknowledge financial support through NSF Grants SES-1824344 and SES-2018498 (Kaido) as well as SES-1824375 (Molinari and Stoye).}}
\date{\today}
\author{Hiroaki Kaido\thanks{Department of Economics, Boston University, hkaido@bu.edu.} 
\and Francesca Molinari\thanks{Department of Economics, 
Cornell University, fm72@cornell.edu.}
\and J\"{o}rg Stoye\thanks{Department of Economics, 
Cornell University, stoye@cornell.edu.}}

\maketitle
\begin{abstract}
The literature on stochastic programming typically restricts attention to problems that fulfill constraint qualifications. The literature on estimation and inference under partial identification frequently restricts the geometry of identified sets with diverse high-level assumptions. These superficially appear to be different approaches to closely related problems. We extensively analyze their relation. Among other things, we show that for partial identification through pure moment inequalities, numerous assumptions from the literature essentially coincide with the Mangasarian-Fromowitz constraint qualification. This clarifies the relation between well-known contributions, including within econometrics, and elucidates stringency, as well as ease of verification, of some high-level assumptions in seminal papers.      
\end{abstract}

\vfill

\pagebreak
\onehalfspacing

\section{Introduction}
\label{sec:introduction}

This paper connects two related but largely separate literatures, namely statistical analysis of stochastic programs and estimation and inference for (functions of) partially identified parameters. The bounds that are pervasive in the latter literature are often expressed as values of constrained optimization problems. As such, some similarity to stochastic programming is rather apparent and has been observed before. However, we uncover much deeper connections between these literatures.

Our discussion starts from the econometrics literature. In a seminal paper, \cite[][CHT henceforth]{Chernozhukov_Hong_etal2007aE} provide a comprehensive analysis of consistency of criterion function-based set estimators and their convergence rates in partially identified models.
Their work highlights the challenges a researcher faces in this context and puts forward possible solutions in the form of assumptions under which specific rates of convergence attain.
While these assumptions can be dispensed with when the researcher's goal is to obtain a confidence set for the partially identified parameter vector that is --pointwise or uniformly-- consistent in level (e.g., \cite{AndrewsSoares2010E}), related assumptions reappear when the aim is to obtain a confidence interval for a smooth function of the partially identified parameter vector that is --pointwise or uniformly-- consistent in level (e.g., \cite[][PPHI henceforth]{PakesPorterHo2011}, \cite[][BCS henceforth]{BCS14_subv}).\footnote{We cite \cite{PakesPorterHo2011} because the published version \citep{PPHI_ECMA} does not contain the inference procedure. However, this procedure has been used in influential papers \citep{Eizenberg,HoPakes14,Holmes11}.} Some more recent contributions \citep{ChoRussell,Gafarov} observe a connection to stochastic programming and show that inference becomes much more tractable under the so-called Linear Independence constraint qualification. Some obvious questions arise: How do all these assumptions relate? What is the trade-off between them and possibly other assumptions from the statistics literature?

For a sense of why consistent estimation of identified sets or their projections can be hard even in otherwise well-behaved moment inequality settings, consider Figure \ref{fig:support point}. Both panels illustrate a detail of an \textit{identified set} $\Theta_I$ defined as the collection of parameter values satisfying a finite number of moment inequalities (as in \eqref{def_theta_i} below, but without equality constraints). Each restriction is represented by a curve and excludes everything above that curve. The resulting $\Theta_I$ is shaded.\footnote{An exact algebraic description of the example in Figure \ref{fig:support point} is as follows. Let $\Theta=\mathbb{R}^2$ with typical element $\theta=(\theta_1,\theta_2)$ and let $\Theta_I=\{\theta:(\theta_2-1)^3+\theta_1E(X_1)\leq 0,(\theta_2-1)^3-\theta_1E(X_1)\leq 0\}$ for the left panel; for the right panel, additionally require $\theta_2\leq E(X_3)$, where $E(X_1)=E(X_2)=E(X_3)=1$. Looking ahead, the right panel of Figure \ref{fig:cones} and both panels of Figure \ref{fig:examples2} are qualitative representations.} Suppose now that a researcher wants to find either (i) a set estimator of $\Theta_I$ that is consistent in Hausdorff distance (defined later) or (ii) an estimator of a linear projection of $\Theta_I$ (e.g., the identified set for a component of $\theta$), represented through the support function
\begin{equation}
s(p,\Theta_I) \equiv \max\{p'\theta:\theta \in \Theta_I\} \label{eq:define_s}
\end{equation}
in pre-determined direction $p$. (In the figure, $p=(0,1)$, i.e. we maximize $\theta_2$.) Even if all constraints' graphs can be estimated at a specific --e.g., parametric-- rate, it does not follow that their intersection estimates either object of interest at the same or indeed at any rate. For example, if estimators approximate the true constraints from below, the support function may be underestimated including in the limit.

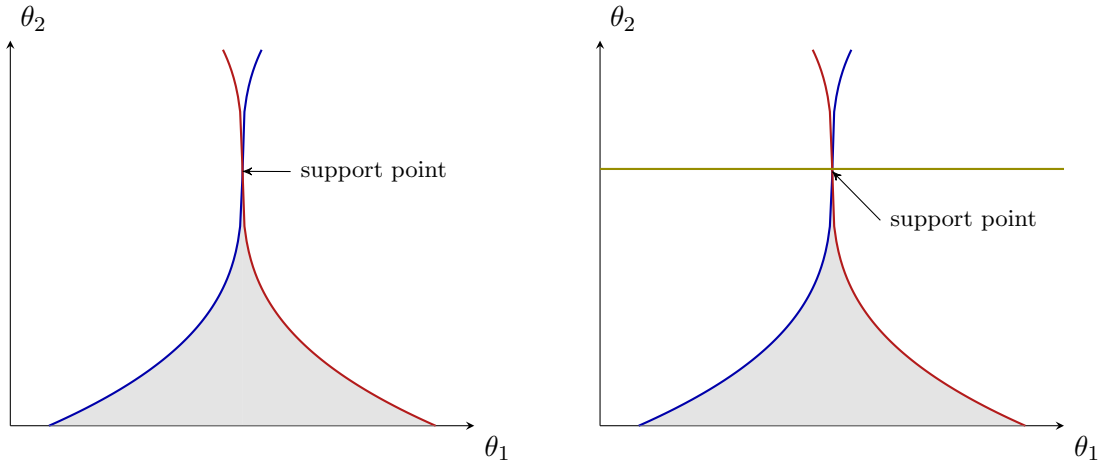
\begin{figure}[t]
		\begin{subfigure}[b]{.5\linewidth}
			\begin{adjustbox}{max width=\textwidth}
  \begin{tikzpicture}[scale=.9]
    \begin{axis}[ticks=none,
            axis lines=middle,clip=false,disabledatascaling,
            xmin=-1.2,xmax=1.2,ymin=0,ymax=1.5,
            axis lines=left,
            xlabel=$\theta_1$,
            ylabel=$\theta_2$,
            xlabel style={at={(ticklabel* cs:1)},anchor=north west},
             ylabel style={at={(ticklabel* cs:1)},anchor=south west},
            ]
      \addplot[name path=U,domain=-1:.1,cornellblue,samples=100,thick] {x/abs(x)*abs(x)^(1/3)+1};
      \addplot[name path=L,domain=-.1:1,cornellred,samples=100,thick] {-x/abs(x)*abs(x)^(1/3)+1} ;
      \addplot[name path=B,domain=-1:1,darkgray!35,opacity=.4,samples=100] {0} ;
      \path [
            name intersections={
                of=U and L,
                by={H},
            },
        ];
        \addplot[fill=darkgray!35,opacity=.4] fill between[of=U and B,soft clip={domain=-1:0}];
      \addplot[fill=darkgray!35,opacity=.4] fill between[of=L and B,soft clip={domain=0:1}];
         \draw[<-] (H) -- (0.25,.99) node[anchor=west]{\footnotesize{support point}};
             \end{axis}
  \end{tikzpicture}
\end{adjustbox}
		\end{subfigure}
		\begin{subfigure}[b]{.5\linewidth}
		\begin{adjustbox}{max width=\textwidth}
 \begin{tikzpicture}[scale=.9]
    \begin{axis}[ticks=none,
            axis lines=middle,clip=false,disabledatascaling,
            xmin=-1.2,xmax=1.2,ymin=0,ymax=1.5,
            axis lines=left,
            xlabel=$\theta_1$,
            ylabel=$\theta_2$,
            xlabel style={at={(ticklabel* cs:1)},anchor=north west},
             ylabel style={at={(ticklabel* cs:1)},anchor=south west},
            ]
      \addplot[name path=U,domain=-1:0.1,cornellblue,samples=100,thick] {x/abs(x)*abs(x)^(1/3)+1};
      \addplot[name path=L,domain=-0.1:1,cornellred,samples=100,thick] {-x/abs(x)*abs(x)^(1/3)+1} ;
      \addplot[domain=-1.2:1.2,olive,samples=100,thick] {1};
      \addplot[name path=B,domain=-1:1,darkgray!35,opacity=.4,samples=100] {0} ;
      \path [
            name intersections={
                of=U and L,
                by={H},
            },
        ];
     \addplot[fill=darkgray!35,opacity=.4] fill between[of=U and B,soft clip={domain=-1:0}];
      \addplot[fill=darkgray!35,opacity=.4] fill between[of=L and B,soft clip={domain=0:1}];
         \draw[<-] (H) -- (0.25,.80) node[anchor=west]{\footnotesize{support point}};
      \end{axis}
  \end{tikzpicture}
\end{adjustbox}
\end{subfigure}
\caption{Two examples of irregular support points. Each curve represents an inequality constraint that excludes everything above it, so that $\Theta_I$ is the shaded region. Estimation of $\Theta_I$ or its support function in direction $(0,1)$ (``up") will be difficult. Note that individual constraints may be well-behaved. The right panel differs from the left one by the addition of a constraint which ensures that a polynomial minorant condition holds. This illustrates that such a condition does not rule out the problem.}
\label{fig:support point}
\end{figure}

The examples may appear ``knife-edge." However, note that: (i) While we will, in this paper, take a pointwise perspective to simplify the analysis, the literature on partial identification is usually concerned with inference that is uniformly valid \textit{near} such irregular cases because asymptotic approximations may otherwise be misleading. Indeed, this is emphasized in the abstract of \cite{CS17}. (ii) Inference methods that are uniformly valid typically use ``Generalized Moment Selection" methods (see again \cite{CS17} for details) that account for statistical uncertainty not only of moment conditions that are violated in sample, but also of ones that are local-to-binding. In these methods, the ``overidentified" feature of the right-hand panel, i.e. the intersection of more than $d$ constraints at one point in $\mathbb{R}^d$, becomes typical of bootstrap d.g.p.'s. In addition, this feature characterizes the boundary case of overidentifying moment conditions, though some assumptions discussed below will exclude that case anyway.  

A reader familiar with constraint qualifications in optimization problems will recognize that both panels of Figure \ref{fig:support point} violate some of these qualifications. Conversely, a reader who is very familiar with the partial identification literature might recognize that they violate assumptions in CHT, PPHI, and elsewhere. We ask if this reflects deeper relations between these bodies of literature. The answer will be affirmative: Under a background assumption of continuous differentiability of moment conditions and abstracting from details of ``how uniformly" the assumptions are stated, the literature on partial identification already invokes constraint qualifications; for examples, we show that both papers just cited rely on the Mangasarian-Fromowitz constraint qualification. This implies some previously unrecognized (to our knowledge) logical relations between assumptions made in econometrics.

Some references to the literatures that we connect are as follows. \cite{MolinariHOE} gives a current overview over the field of partial identification. \cite{CS17} provide a definitive treatment of the literature on moment inequalities. We will define constraint qualifications below but refer to \cite{baz:she:she06} for a textbook treatment and to \cite{BonnansShapiroBook} for a textbook on perturbation analysis of stochastic programs.
   
\section{Assumptions} \label{sec:assumptions}
This section clarifies the general setup and introduces a broad array of assumptions from the literature. For the reader's convenience, these assumptions, and some results that we will explain later, are summarized in Table \ref{tab:assumptions}.

We assume throughout that the model is correctly specified and that individual moment conditions are well-behaved. Specifically, the identified set $\Theta_I$ is characterized by $J$ constraints, namely $J_1\le J$ moment inequalities and $J_2\equiv J-J_1$ moment equalities:
\begin{equation}
\Theta_I \equiv \{\theta \in \Theta: E(m_j(X,\theta))\le 0,~j=1,\dots,J_1;E(m_j(X,\theta))= 0,~j=J_1+1,\dots,J\}, \label{def_theta_i}
\end{equation}
where the functions $(m_1(\cdot),\dots,m_J(\cdot))$ are known up to $\theta \in \Theta$. Then we impose:
\begin{assumption} \label{as:momP_AS}

(a) $\Theta\subset\mathbb{\ R}^d$ is compact convex with nonempty interior. 

(b) $\Theta_I \neq \emptyset$ and $\Theta_I \subset \operatorname{int}(\Theta)$.

(c) $\sigma^2_j(\theta)\equiv \operatorname{Var}(m_j(X,\theta)) \in(0,\infty)$ for $j=1,\dots, J$.

(d)  The gradients $D_j(\cdot)\equiv\nabla_\theta\{E[m_j(X,\cdot)]/\sigma_j(\cdot)\},j=1,\dots,J$ exist and are continuous.
\end{assumption}
These assumptions are standard in the literature. The requirement that $\Theta_I \subset \operatorname{int}(\Theta)$ may appear stronger than the comparable one in CHT, i.e. their condition M2. However, the latter imposes continuous differentiability of moment conditions on a small enlargement of $\Theta$, so it does constrain $\Theta_I$ to be interior to the set on which local linear approximation of moment conditions is valid. This is how we use the assumption, and we could analogously weaken it.  

We next state numerous assumptions that are inspired by the aforementioned literature in econometrics. We first state them in a way that maximizes resemblance to the original formulation, subject to the unification that assumptions are stated pointwise (not uniformly) over d.g.p.'s and that their local implications near extreme points of $\Theta_I$ are extracted. Universal constants invoked in assumptions need not take the same value across appearances.

Define the support set of $\Theta_I$ as
$$ S(p,\Theta_I) \equiv  \{\theta \in \Theta_I: p'\theta=s(p,\Theta_I)\},$$
where the support function $s(\cdot)$ is defined in \eqref{eq:define_s}, and the supporting hyperplane as
$$ H(p,\Theta_I) \equiv  \{\theta \in \Theta: p'\theta=s(p,\Theta_I)\}.$$
Any element of $S(p,\Theta_I)$ is also called a support point. We use $\theta^*$ to denote a generic support point; to economize on subscripts and because we consider $p$ fixed, we suppress dependence of $\theta^*$ on $p$. Recall also that a constraint is \textit{active} at $\theta$ if $E(m_j(X,\theta))=0$. Let
\begin{equation}
\cJ_1^*(\theta)\equiv \left\{j \in \{1,\dots,J_1\}:E(m_j(X,\theta))=0\right\}\label{eq:J1star}
\end{equation}
denote the set of active inequality constraints and
\begin{equation}
\cJ^*(\theta)\equiv \left\{j:E(m_j(X,\theta))=0\right\}=\cJ_1^*(\theta) \cup \{J_1+1,\dots,J\}\label{eq:Jstar}
\end{equation}
the full \textit{active set} of (equality or inequality) constraints at $\theta$.

We first adapt two assumptions, ``Degeneracy" and ``Polynomial Minorant," from CHT. These are essential for getting rate results for consistency of analog estimators of $\Theta_I$; in particular, Polynomial Minorant ensures rate results for a relaxed sample analog of $\Theta_I$, and Degeneracy allows one to drop the relaxation. We weaken the assumptions insofar as they are only imposed at support points. Also, we do not adapt the high-level Conditions C.2 and C.3 from CHT because, being about sample objects, they restrict the sampling process and not just population moments. To keep these issues separate, we focus on the sufficient conditions that stand in for the assumptions in a moment inequalities setting (i.e., CHT's displays 4.5 and 4.6).

\begin{assumption}\label{as:CHT-degeneracy}

For each support point $\theta^*\in S(p,\Theta_I)$, there exist constants $\delta>0$, $\eta>0$, $M>0$, and $C>0$ s.t.
\begin{gather}
\max_{j=1,\dots,J} E(m_j(X,\theta))/\sigma_j(\theta) \leq -C \varepsilon ~~ \text{for all} ~ \theta \in \Theta_I^{-\varepsilon} \cap B(\theta^*,\eta) \label{eq:CHT-deg-1} \\
\max_{\theta \in \Theta_I \cap B(\theta^*,\eta)} d(\theta,\Theta_I^{-\varepsilon}) \leq M\varepsilon ~~ \text{for all} ~ \varepsilon \in [0,\delta], \label{eq:CHT-deg-2}
\end{gather}
where $\Theta_I^{-\varepsilon}=\{\theta \in \Theta_I:d(\theta,\Theta \setminus \Theta_I) \geq \varepsilon\}$ and $B(\theta^*,\eta)\equiv \{\theta \in \Theta: \Vert \theta-\theta^* \Vert  \leq \eta \}$.
\end{assumption}
This ``degeneracy" assumption ensures that for any support point $\theta^*$, there exists a nearby point where moment inequalities hold with slack and whose membership in $\Theta_I$ is therefore easy to determine. In particular, we show in the proof of Theorem \ref{thm:main} that Assumption \ref{as:CHT-degeneracy} precludes the existence of equality constraints.

In the original version, the equivalent of \eqref{eq:CHT-deg-2} was stated using Hausdorff distance, i.e. $d_H(\Theta_I,\Theta_I^{-\varepsilon}) \leq M\varepsilon$, where $d_H(A,B)=\max\{\max_{a \in A} d(a,B),\max_{b \in B} d(b,A)\}$ for generic sets $A,B$ with typical elements $a,b$. However, $\max_{\theta \in \Theta_I^{-\varepsilon}} d(\theta,\Theta_I)=0$ because $\Theta_I^{-\varepsilon} \subset \Theta_I$, so here and in the original, only the implied restriction on $\max_{\theta \in \Theta_I} d(\theta,\Theta_I^{-\varepsilon})$ is nonvacuous. We localize it by restricting $\theta$ to a neighborhood of the support point.

\begin{assumption}\label{as:CHT-PM}

For each support point $\theta^*\in S(p,\Theta_I)$, there exist constants $\eta>0$, $c>0$, and $C>0$ such that,  for all $\theta \in  B(\theta^*,\eta)$,
$$ \max\bigl\{0,\max_{j=1,\dots,J_1}E(m_j(X,\theta))/\sigma_j(\theta),\max_{j=J_1+1,\dots,J}\vert E(m_j(X,\theta))/\sigma_j(\theta)\vert\bigr\} \geq C \times \min\{d(\theta,\Theta_I),c\}. $$
\end{assumption}
This ``minorant" assumption ensures that the population criterion increases not too slowly as one moves away from $\Theta_I$. Loosely speaking, it prevents ``weak identification" problems.
Next, BCS impose a polynomial minorant condition as well.\footnote{In the respective originals, both assumptions raise the r.h.s. to a power $\gamma>0$, hence ``polynomial minorant." However, in both cases, $\gamma$ is restricted in ways that imply $\gamma=1$ in the present setting.}  

\begin{assumption}\label{as:BCS-PM}

For each support point $\theta^*\in S(p,\Theta_I)$, there exist constants $\eta>0$, $c>0$, and $C>0$ such that, for all $\theta \in  B(\theta^*,\eta) \cap H(p,\Theta_I)$,  
$$ \max\bigl\{0,\max_{j=1,\dots,J_1}E(m_j(X,\theta))/\sigma_j(\theta),\max_{j=J_1+1,\dots,J}\vert E(m_j(X,\theta))/\sigma_j(\theta)\vert\bigr\}\geq C \times \min\{d(\theta,S(p,\Theta_I)),c\} . $$
\end{assumption}
It bears emphasis that the last two assumptions are logically independent. Assumption \ref{as:CHT-PM} forces the population criterion to increase (to first order) as we move away from $\Theta_I$. Assumption \ref{as:BCS-PM} enforces an analogous increase as we move away from the support set $S(p,\Theta_I)$ \textit{along the supporting hyperplane}. This does not imply Assumption \ref{as:CHT-PM} because it only applies as one leaves $\Theta_I$ in selective directions. But it also is not implied because the directions considered in Assumption \ref{as:BCS-PM} may be tangential to $\Theta_I$, in which case the increase in $d(\theta,\Theta_I)$ may be of low order. Indeed, Assumption \ref{as:BCS-PM} may be considered restrictive: By not allowing directions to be tangential to $\Theta_I$ without also being tangential to $S(p,\Theta_I)$, it excludes smooth maxima, e.g. any identified set whose entire boundary is a smooth manifold.\footnote{Consider the single moment inequality $\theta_1^2+\theta_2^2-E(X) \leq 0$, with $E(X)=1$ and $p=(0,1)$. Then $S(p,\Theta_I)$ is a singleton at $\theta^*=(0,1)$. Let $\theta=(\zeta,1)$, then $\theta_1^2+\theta_2^2-E(X)=\zeta^2$, whereas $d(\theta,S(p,\Theta_I))=\zeta$, violating Assumption \ref{as:BCS-PM} as $\zeta \to 0$; at the same limit, we have $d(\theta,\Theta_I)=\sqrt{1+\zeta^2}-1>\zeta^2/4$, so that Assumption \ref{as:CHT-PM} holds. We note that this and other easy counterexamples to Assumption \ref{as:BCS-PM} do not seem to be counterexamples to the BCS profiling method. We leave to future research the question whether this illustrates the possibility of relaxing their assumptions.}

We next adapt (in this order) assumptions 3, 4(a), and 4(b) from PPHI. These are modified in a few ways: PPHI assume that the support point $\theta^*$ is unique; this assumption is removed. Assumptions are also localized by only looking at $\theta$ near $\theta^*$; this makes the first of them meaningfully weaker. At the same time, PPHI impose assumptions uniformly over d.g.p.'s; to keep this paper focused, such uniform statements are removed throughout.\footnote{To compare the statements of assumptions, also keep in mind notational conventions: PPHI write moment conditions as $E(m_j(\cdot))\geq 0$, set $p=(-1,0,\dots,0)$, and use $1$-subscripts to denote first components of vectors, so that $t_1$ there would be $-p't$ here.}
%\newpage
\begin{assumption}\label{as:PPHI-cone}

For each support point $\theta^*\in S(p,\Theta_I)$, there exist $\delta,\eta>0$ s.t. $\sup_{t \in T(\theta^*,\eta)}p't \leq -\delta$, where $$T(\theta^*,\eta) \equiv \left\{t=\frac{\theta-\theta^*}{\Vert \theta-\theta^*\Vert}:\theta \in \Theta_I \cap B(\theta^*,\eta),\theta \neq \theta^* \right\}.$$ 
\end{assumption}

\begin{assumption}\label{as:PPHI-descent}

There are no equality constraints, i.e., $J_1=J$. Furthermore, for each support point $\theta^*\in S(p,\Theta_I)$, there exist constants $\delta,\varepsilon>0$ as well as direction (i.e. unit vector) $t \in \R^d$ s.t. $$\max_{j:E(m_j(X,\theta^*))/\sigma_j(\theta^*) \geq -\delta}D_j(\theta^*)t \leq -\varepsilon.$$
\end{assumption}

\begin{assumption}\label{as:PPHI-ascent}

For each support point $\theta^*\in S(p,\Theta_I)$, there exist constants $\delta<0$, $\varepsilon>0$ s.t. $$ \inf_{t \in \mathbb{R}^d,\Vert t \Vert=1,p't \geq \delta} \max\bigl\{ \max_{j \in \mathcal{J}_1^*(\theta^*)} D_j(\theta^*)t, \max_{j \in \{J_1+1,...J\} } \vert D_j(\theta^*)t\vert \bigr\} > \varepsilon.$$
\end{assumption}

Brief intuitions for these are as follows. Assumption \ref{as:PPHI-cone} ensures that $\Theta_I$ is contained in a cone that has $\theta^*$ as apex and does not otherwise intersect $H(p,\Theta_I)$. In particular, this implies uniqueness of $\theta^*$ (although we will not use this feature) and pointiness of the tangent cone (defined later) at $\theta^*$. Assumption \ref{as:PPHI-descent} ensures that locally to $\theta^*$, there exists a direction in which all moment expectations, hence their maximum, strictly decrease. Note in particular that this excludes equality constraints -- while this implication is not explicit in PPHI, they treat equalities as conjunctions of two inequalities, and the assumption cannot possibly hold for both. PPHI point out that their assumption is inspired by CHT's Degeneracy (i.e., our Assumption \ref{as:CHT-degeneracy}), which also excludes equalities. Indeed, both assumptions force $\Theta_I$ to have an interior, and we will have more to say about their relation later. Assumption \ref{as:PPHI-ascent} enforces that the criterion is strictly increasing in many directions from $\theta^*$, including all that have positive inner product, and some that have  negative inner product, with $p$.\footnote{The restriction $\delta<0$ in Assumption \ref{as:PPHI-ascent} correctly reflects our source. However, we considered the possibility that (in our notation) $\delta>0$ was intended. The assumption then becomes weaker. Specifically, along the lines of our main result below, one can show that it is then equivalent to $\max_{j \in \mathcal{J}^*(\theta^*)} D_j(\theta^*)p >0$ and is implied by Assumption \ref{ACQ}.}

We finally recall some classic constraint qualifications. This requires some standard notation that we will also use later. 
For any $\theta \in \Theta_I$, define the tangent cone
$$ \mathcal{T}(\theta) \equiv \biggl\{t \in \mathbb{R}^d: \exists \{\theta_m\}_{m=1}^{\infty} \subset \Theta_I, \theta_m \to \theta, \lim_{m \to \infty} \frac{\theta_m-\theta}{\Vert \theta_m-\theta \Vert} = \frac{t}{\Vert t \Vert} \biggr\} \cup \{0\} $$
as well as the linearized cone
$$ \mathcal{L}(\theta) \equiv \bigl\{t \in \mathbb{R}^d: D_j(\theta)t \leq 0,j \in \cJ_1^*(\theta);D_j(\theta)t = 0,j \in \{J_1+1,\dots,J\}  \bigr\}. $$
We will only invoke these objects for support points $\theta^* \in S(p,\Theta_I)$. Recall that $\mathcal{T}(\theta) \subseteq \mathcal{L}(\theta)$  \citep[][chapter 5]{baz:she:she06}.
 
Both cones are illustrated in Figure \ref{fig:cones}. They coincide in ``nice" examples like the right panel of Figure \ref{fig:cones} (indeed, this ``niceness" is the Abadie constraint qualification defined below) but they may disagree, as in the left panel which reprises the left panel from Figure \ref{fig:support point}.

\begin{figure}[t]
		\begin{subfigure}[b]{.5\linewidth}
			\begin{adjustbox}{max width=\textwidth}
  \begin{tikzpicture}[scale=.9]
    \begin{axis}[ticks=none,
            axis lines=middle,clip=false,disabledatascaling, 
            xmin=-1.2,xmax=1.2,ymin=0,ymax=1.5,
            axis lines=left,
            xlabel=$\theta_1$,
            ylabel=$\theta_2$,
            xlabel style={at={(ticklabel* cs:1)},anchor=north west},
             ylabel style={at={(ticklabel* cs:1)},anchor=south west},
            ]
             \addplot[thick,orange] coordinates{(0,0) (0,1.47)};
              \addplot[thick,green,dashed] coordinates{(0,0) (0,1)};
      \addplot[name path=U,domain=-1:.1,samples=100,thick] {x/abs(x)*abs(x)^(1/3)+1};
      \addplot[name path=L,domain=-.1:1,samples=100,thick] {-x/abs(x)*abs(x)^(1/3)+1} ;
      \addplot[name path=B,domain=-1:1,darkgray!35,opacity=.4,samples=100] {0} ;
      \path [
            name intersections={
                of=U and L,
                by={H},
            },
        ];
     \addplot[fill=darkgray!35,opacity=.4] fill between[of=U and B,soft clip={domain=-1:0}];
      \addplot[fill=darkgray!35,opacity=.4] fill between[of=L and B,soft clip={domain=0:1}];
      \draw[<-] (H) -- (0.25,.99) node[anchor=west]{\footnotesize{support point}};
              \end{axis}
\end{tikzpicture}
\end{adjustbox}
		\end{subfigure}
		\begin{subfigure}[b]{.5\linewidth}
		\begin{adjustbox}{max width=\textwidth}
 \begin{tikzpicture}[scale=.9]
    \begin{axis}[ticks=none,
            axis lines=middle,clip=false,disabledatascaling, 
            xmin=-1.2,xmax=1.2,ymin=0,ymax=1.5,
            axis lines=left,
            xlabel=$\theta_1$,
            ylabel=$\theta_2$,
            xlabel style={at={(ticklabel* cs:1)},anchor=north west},
             ylabel style={at={(ticklabel* cs:1)},anchor=south west},
            ]
      \addplot[domain=-7/12:7/12,thick,orange] {1-12/7*abs(x)};
      \addplot[domain=-7/12:7/12,thick,green,dashed] {1-12/7*abs(x)};
      \addplot[name path=U,domain=-1:0.2,samples=100,thick] {-1/35+5/7*(x+1.2)^2};
      \addplot[name path=L,domain=-.2:1,samples=100,thick] {-1/35+5/7*(x-1.2)^2};
      \addplot[name path=B,domain=-1:1,darkgray!35,opacity=.4,samples=100] {0} ;
            \path [
            name intersections={
                of=U and L,
                by={H},
            },
        ];
     \addplot[fill=darkgray!35,opacity=.4] fill between[of=U and B,soft clip={domain=-1:0}];
      \addplot[fill=darkgray!35,opacity=.4] fill between[of=L and B,soft clip={domain=0:1}];
     \draw[<-] (H) -- (0.25,1) node[anchor=west]{\footnotesize{support point}};
%     \draw [bend left,green,very thick] (-0.15,0.7589)--(0.15,0.7589) ;
	 \centerarc[green, thick](-0,1)(240:300:.3);
 \centerarc[orange, thick](-0,1)(240:300:.32);
      \end{axis}
  \end{tikzpicture}
\end{adjustbox}
\end{subfigure}
\caption{An illustration of linear cone (orange) versus tangent cone (green). Left panel: The two disagree in the example from Figure \ref{fig:support point}; note that the linear cone goes both ``up" and ``down" from the support point. Right panel: The two agree in a well-behaved setting that cannot be statistically distinguished from the ill-behaved one.}
\label{fig:cones}
\end{figure}
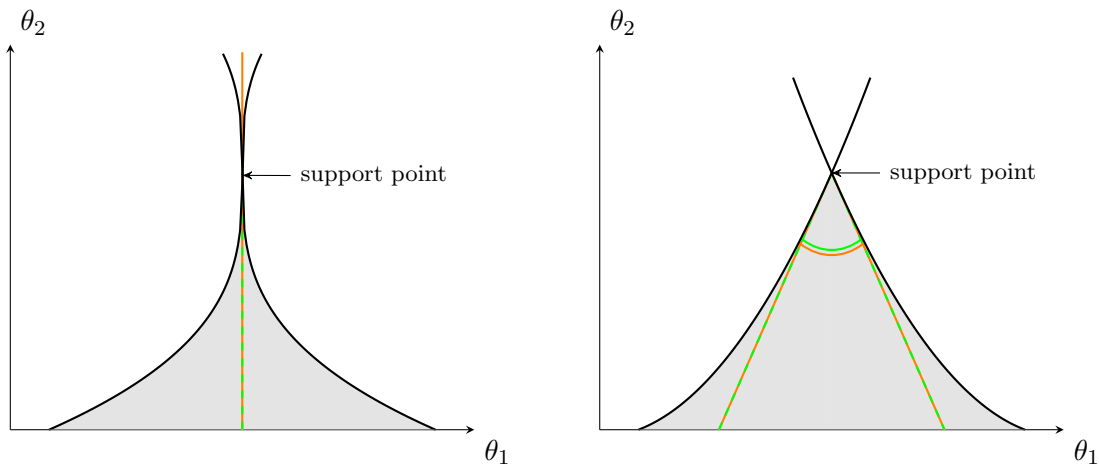

We can then state the following constraint qualifications in decreasing order of restrictiveness. 

\begin{assumption}{\textbf{Linear Independence Constraint Qualification} (LICQ)} \label{LICQ}

For each support point $\theta^* \in S(p,\Theta_I)$, the active (equality or inequality) constraints have linearly independent gradients $D_j(\theta^*)$.
\end{assumption}

Of course, the LICQ requires at most $d$ active constraints, making it quite restrictive.

\begin{assumption}{\textbf{Mangasarian-Fromowitz Constraint Qualification} (MFCQ)} \label{MFCQ}

At each support point $\theta^* \in S(p,\Theta_I)$, the gradients of the equality constraints are linearly independent and there exists $t \in \mathbb{R}^d$ s.t. $D_j(\theta^*)t<0$ for $j \in \cJ_1^*(\theta^*)$ and $D_j(\theta^*)t=0$ for $j \in \{J_1+1,\dots,J\}$.
\end{assumption}
\newpage

\begin{assumption}{\textbf{Abadie Constraint Qualification} (ACQ)}\label{ACQ}

For each support point $\theta^* \in S(p,\Theta_I)$, we have $\mathcal{L}(\theta^*)=\mathcal{T}(\theta^*)$.
\end{assumption}
 
These conditions are frequently invoked in the statistical literature. For example, \cite{Shapiro90,Shapiro1991aAOR,Shapiro93} uses either LICQ or uniqueness of Lagrange multipliers, and these two assumptions are essentially the same \citep{Wachsmuth}. In econometrics, \cite{ChoRussell} and \cite{Kaido:Santos} use LICQ; \cite{andrews_roth_pakes} and \cite{Gafarov} restrict attention to linear constraints and thereby impose ACQ.

\section{Results}

\subsection{Restating Some Assumptions, and a First Set of Equivalences}

We next restate some of these assumptions, exploiting their localization or using the language of constraint qualifications. Between our ``localization" of assumptions and the language of tangent cones, several of the assumptions we just introduced can be stated more succinctly. 
In what follows, recall that $\cJ^*(\theta)$, defined in \eqref{eq:Jstar}, is the active set of (equality or inequality) constraints at $\theta$.
Specifically, define:

\begin{customassumption}{3'}\label{as:CHT-PM_2}
For each support point $\theta^*\in S(p,\Theta_I)$, there exist constants $\eta>0$ and $C>0$ such that,  for all $\theta \in  B(\theta^*,\eta)$,
$$ \max\bigl\{0,\max_{j=1,\dots,J_1}E(m_j(X,\theta))/\sigma_j(\theta),\max_{j=J_1+1,\dots,J}\vert E(m_j(X,\theta))/\sigma_j(\theta)\vert\bigr\} \geq C \times d(\theta,\Theta_I). $$
\end{customassumption}

\begin{customassumption}{4'}\label{as:BCS-PM_2}

For each support point $\theta^*\in S(p,\Theta_I)$, there exist constants $\eta>0$ and $C>0$ such that, for all $\theta \in  B(\theta^*,\eta) \cap H(p,\Theta_I)$,  
$$ \max\bigl\{0,\max_{j=1,\dots,J_1}E(m_j(X,\theta))/\sigma_j(\theta),\max_{j=J_1+1,\dots,J}\vert E(m_j(X,\theta))/\sigma_j(\theta)\vert\bigr\}\geq C \times d(\theta,S(p,\Theta_I)) . $$
\end{customassumption}

\begin{customassumption}{5'} \label{as:PPHI-cone_2}

For each support point $\theta^*\in S(p,\Theta_I)$,
$\max\{p't/\Vert t \Vert: t\in \mathcal{T}(\theta^*) \setminus \{0\}\}<0$.
\end{customassumption}

\begin{customassumption}{6'}\label{as:PPHI-descent_2}

There are no equality constraints. Furthermore, for each support point $\theta^*\in S(p,\Theta_I)$, $\min_{t \in \mathbb{R}^d} \max_{j \in \cJ^*(\theta^*)}D_j(\theta^*)t/\Vert t \Vert <0$.
\end{customassumption}

\begin{customassumption}{7'}\label{as:PPHI-ascent_2}

For each support point $\theta^*\in S(p,\Theta_I)$,
$\max\{p't/\Vert t \Vert: t\in \mathcal{L}(\theta^*) \setminus \{0\}\}<0$.
\end{customassumption}

Then we have:

\begin{lemma}\label{lem:restate}
Suppose that Assumption \ref{as:momP_AS} holds. Then the following assumptions are equivalent: \ref{as:CHT-PM}$\Leftrightarrow$\ref{as:CHT-PM_2}, \ref{as:BCS-PM}$\Leftrightarrow$\ref{as:BCS-PM_2}, \ref{as:PPHI-cone}$\Leftrightarrow$\ref{as:PPHI-cone_2}, \ref{as:PPHI-descent}$\Leftrightarrow$\ref{as:PPHI-descent_2}, and \ref{as:PPHI-ascent}$\Leftrightarrow$\ref{as:PPHI-ascent_2}.
\end{lemma}
\begin{proof}
Regarding Asumptions \ref{as:CHT-PM} and \ref{as:CHT-PM_2}, $\Leftarrow$ is obvious and $\Rightarrow$ holds because in Assumption \ref{as:CHT-PM}, one can choose $\eta=c$, ensuring $\min\{d(\theta,\Theta_I),c\}=d(\theta,\Theta_I)$. The argument for \ref{as:BCS-PM}$\Leftrightarrow$\ref{as:BCS-PM_2} is the same. 

Next, \ref{as:PPHI-cone}$\Rightarrow$\ref{as:PPHI-cone_2} holds because $T(\theta^*,\eta)$ shrinks toward $\cT(\theta^*)\setminus \{0\}$ as $\eta \to 0$. Also, suppose $\cT(\theta^*)\setminus \{0\}=\emptyset$, then Assumption \ref{as:PPHI-cone_2} holds vacuously; but in this case, $T(\theta^*,\eta)=\emptyset$ for small enough $\eta$ and so Assumption \ref{as:PPHI-cone} holds as well. It remains to show that, if $\cT(\theta^*)\setminus \{0\}\neq \emptyset$ and therefore $T(\theta^*,\eta)\neq \emptyset$ for any $\eta$, then failure of Assumption \ref{as:PPHI-cone} implies failure of Assumption \ref{as:PPHI-cone_2}. Now, failure of Assumption \ref{as:PPHI-cone} and nonemptiness of $T(\theta^*,\eta)$ jointly imply existence of sequences $(\delta_n,\eta_n)\downarrow (0,0)$ and $\theta_n \in B(\theta^*,\eta_n) \setminus \{\theta^*\}$ s.t. $p'(\theta_n-\theta^*)/\Vert \theta_n-\theta^* \Vert>-\delta_n$. But then any accumulation point $t$ of $(\theta_n-\theta^*)/\Vert \theta_n-\theta^* \Vert$ is in $\cT(\theta^*)$ and has $p't \geq 0$, contradicting Assumption \ref{as:PPHI-cone_2}.

Assumption \ref{as:PPHI-descent} obviously implies \ref{as:PPHI-descent_2}. To see the reverse implication, suppose \ref{as:PPHI-descent_2} holds, then one can verify Assumption \ref{as:PPHI-descent} by choosing $\delta$ to be half the slack of the tightest inactive inequality (or arbitrarily if all inequalities are active).

Next, note first that
\begin{eqnarray}\label{eq:conetransform}
&& \max\bigl\{p't/\Vert t \Vert : t\in \mathcal{L}(\theta^*) \setminus \{0\}\bigr\}<0 \\
&\Longleftrightarrow &\min_{t \in \mathbb{R}^d:\Vert t \Vert=1,p't \geq 0}\max\bigl\{ \max_{j \in \mathcal{J}_1^*(\theta^*)} D_j(\theta^*)t, \max_{j \in \{J_1+1,\dots, J\} } \vert D_j(\theta^*)t\vert \bigr\} > 0. \notag
\end{eqnarray}
For example, it is easy to see that the above minimum is attained. If it equalled $0$, the vector $t$ attaining it would be in $\mathcal{L}(\theta^*)$, so the first maximum would be at least $0$. The converse argument is similar.

The left-hand side of \eqref{eq:conetransform} is Assumption \ref{as:PPHI-ascent_2}. We next show that its right-hand side is equivalent to Assumption \ref{as:PPHI-ascent}. It is implied by Assumption \ref{as:PPHI-ascent} because the minimization is over a smaller set. To see the converse, suppose Assumption \ref{as:PPHI-ascent} fails, then there exist sequences $\delta_n \uparrow 0$, $\varepsilon_n \downarrow 0$, and $t_n$ with $p't_n \geq \delta_n$ and $\max\bigl\{ \max_{j \in \mathcal{J}_1^*(\theta^*)} D_j(\theta^*)t_n, \max_{j \in \{J_1+1,\dots, J\} } \vert D_j(\theta^*)t_n\vert \bigr\} \leq \varepsilon_n$. Any accumulation point of $t_n$ then is a counterexample to the right-hand side of \eqref{eq:conetransform}. 
\end{proof}
\begin{remark}
Some of these equivalences are due to localization of assumptions. For example, the original Assumption 3 in PPHI, but also both polynomial minorant conditions, are otherwise stronger than their simplifications. Indeed, regarding the equivalences \ref{as:CHT-PM}$\Leftrightarrow$\ref{as:CHT-PM_2} and \ref{as:BCS-PM}$\Leftrightarrow$\ref{as:BCS-PM_2}, the real insight is that the original assumptions combine local and global conditions, e.g. (for Assumption \ref{as:CHT-PM})
\begin{eqnarray}
&& d(\theta,\Theta_I) \leq \delta\label{eq:CHT-orig-local} \\
&\implies & \max\bigl\{0,\max_{j=1,\dots,J_1}E(m_j(X,\theta))/\sigma_j(\theta),\max_{j=J_1+1,\dots,J}\vert E(m_j(X,\theta))/\sigma_j(\theta)\vert\bigr\} \geq C \times d(\theta,\Theta_I) \notag 
\end{eqnarray}
and
\begin{eqnarray}
&& d(\theta,\Theta_I) \geq \delta \label{eq:CHT-orig-global} \\
&\implies & \max\bigl\{0,\max_{j=1,\dots,J_1}E(m_j(X,\theta))/\sigma_j(\theta),\max_{j=J_1+1,\dots,J}\vert E(m_j(X,\theta))/\sigma_j(\theta)\vert\bigr\} \geq \varepsilon \notag
\end{eqnarray}
for some $\varepsilon>0$ (namely, setting $\varepsilon=C\delta$). Only the local condition \eqref{eq:CHT-orig-local} is related to constraint qualifications. The other part is really a global identification condition. We will revisit this distinction later.
\end{remark}

\subsection{Econometric Assumptions as Constraint Qualifications}
\label{subsec:main_theorem}

We now present our main insight: Several of the above assumptions are equivalent, or very close to, constraint qualification assumptions, and there are numerous logical relations between them. Our main result, which is also visualized in Figure \ref{fig:flow chart}, follows.

\begin{theorem} \label{thm:main}
Suppose Assumption \ref{as:momP_AS} holds. Then the following relations between assumptions hold true.
\begin{enumerate}
\item Assumptions \ref{as:CHT-degeneracy} and \ref{as:PPHI-descent} are equivalent. Furthermore, both are equivalent to jointly (i) excluding equality restrictions and (ii) imposing \ref{MFCQ}.
\item Any of Assumptions \ref{as:CHT-degeneracy}, \ref{as:PPHI-descent}, and \ref{MFCQ} (the latter combined with excluding equality constraints) imply \ref{as:CHT-PM}.
\item Assumption \ref{as:CHT-PM} implies \ref{ACQ}.
\item Assumptions \ref{as:PPHI-cone} and \ref{ACQ} jointly imply \ref{as:PPHI-ascent}.
\item Assumption \ref{as:PPHI-ascent} implies \ref{as:PPHI-cone}.
\item Assumption \ref{as:PPHI-ascent} implies that the gradients of active constraints span $\mathbb{R}^d$. In particular, if $\cJ^*(\theta^*)$ has exactly $d$ elements, \ref{LICQ} is implied.    
\item Assumption \ref{as:PPHI-ascent} implies \ref{as:BCS-PM}.
\end{enumerate}
\end{theorem}

\begin{table}
\begin{tabular}{ll}
\textbf{Assumption} & \textbf{Intuition}  \\
\hline
1 & Background assumption.  \\
2 & All moments strictly negative at some $\tilde{\theta}$ close to $\theta^*$.  \\
3 $\Leftrightarrow$ 3' & Criterion increases with distance from $\Theta_I$.  \\
4 $\Leftrightarrow$ 4' & Criterion increases with distance from support set on support plane.  \\
5 $\Leftrightarrow$ 5' & $\Theta_I$ is locally contained in a pointy cone.  \\
6 $\Leftrightarrow$ 6' & There exists a direction from $\theta^*$ in which every constraint becomes slack.  \\
7 $\Leftrightarrow$ 7' & Criterion increases outside a pointy cone locally containing $\Theta_I$.  \\
8 & Linear Independence Constraint Qualification (LICQ).  \\
9 & Mangasarian-Fromowitz Constraint Qualification (MFCQ).  \\
10 & Abadie Constraint Qualification (ACQ).  \\
\end{tabular}
\caption{Assumptions and their relation as established in Lemma \ref{lem:restate}. The term ``criterion" refers to $\max\bigl\{0,\max_{j=1,\dots,J_1}E(m_j(X,\theta))/\sigma_j(\theta),\max_{j=J_1+1,\dots,J}\vert E(m_j(X,\theta))/\sigma_j(\theta)\vert\bigr\}$ .}
\label{tab:assumptions}
\end{table}

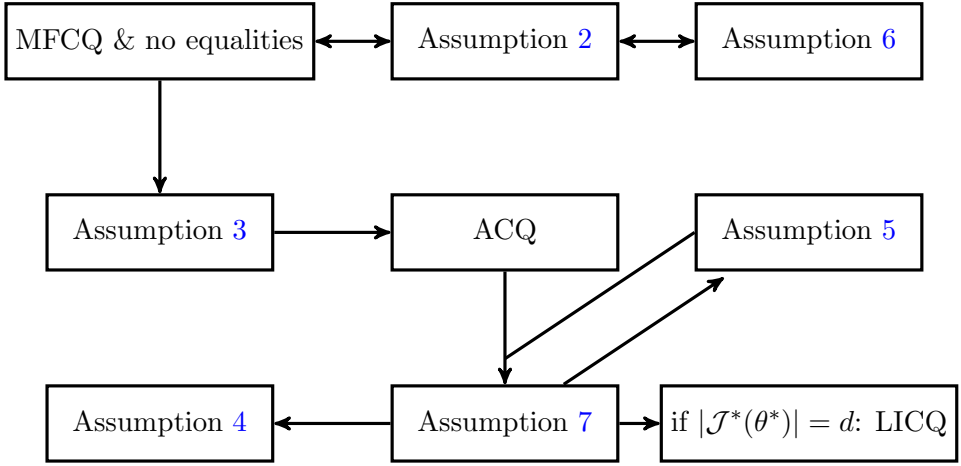
\begin{figure}

\begin{adjustbox}{max width=\textwidth}

%define block styles
\tikzstyle{lemma} = [rectangle, minimum width=3cm, minimum height=1cm,text centered, draw=black, fill=white!100]
\tikzstyle{remark} = [rectangle, minimum width=1cm, minimum height=1cm,text centered, draw=white, fill=white!100]
\tikzstyle{comment} = [rectangle, rounded corners, minimum width=3cm, minimum height=1cm,text centered, draw=black, fill=gray!30]
\tikzstyle{line} = [draw, -latex']

\begin{tikzpicture}[node distance = 2cm, auto]
\linespread{0.9}
% Place nodes
%\node [lemma,very thick] (LICQ) {LICQ};

\node [lemma,very thick] (CHT-degen) {MFCQ \& no equalities};

\node [lemma, right = 1cm of CHT-degen,very thick] (MFCQ) {Assumption \ref{as:CHT-degeneracy}};

\node [lemma, right = 1cm of MFCQ,very thick] (PPHI-unpointy) {Assumption  \ref{as:PPHI-descent}};

\node [lemma, below =1.5cm of MFCQ,very thick] (ACQ) {ACQ};

\node [lemma, below =1.5cm of CHT-degen,very thick] (CHT-LM) {Assumption  \ref{as:CHT-PM}};

\node [lemma, below =1.5cm of PPHI-unpointy,very thick] (PPHI-pointy) {Assumption  \ref{as:PPHI-cone}};

\node [lemma, below =1.5cm of PPHI-pointy, very thick] (LICQ-d) {if $|\cJ^*(\theta^*)|=d$: LICQ};

\node [lemma, below =1.5cm of ACQ,very thick] (PPHI-ascent) {Assumption \ref{as:PPHI-ascent}};

\node [lemma, below =1.5cm of CHT-LM,very thick] (BCS-LM) {Assumption \ref{as:BCS-PM}};

\draw [<->,very thick] (MFCQ) -- (CHT-degen);

\draw [<->,very thick] (MFCQ) -- (PPHI-unpointy);

\draw [->,very thick] (CHT-degen) -- (CHT-LM);

\draw [->,very thick] (CHT-LM) -- (ACQ);

\draw [->,very thick] (ACQ) -- (PPHI-ascent);

%\draw [->,very thick] (PPHI-pointy) -- (LICQ-d);
\draw [->,very thick] (PPHI-ascent) -- (LICQ-d);

%\draw [-,very thick] (PPHI-ascent.north) -- +(0,.76) -- +(4.35,.76);
%\draw [-,very thick] (PPHI-pointy.south) -- +(0,-.7) -- +(-4.35,-.7);

%\draw [<-,very thick] (BCS-LM.north) -- +(0,.76) -- +(3,.76);
\draw [<-,very thick] (BCS-LM) -- (PPHI-ascent);

\node [above=.5cm of PPHI-ascent] (helper) {} ;

\draw [->,very thick] (PPHI-ascent) -- +(2.91,1.94);
\draw [very thick] (PPHI-pointy.west) -- +(-2.52,-1.68);

\end{tikzpicture}
\end{adjustbox}
\caption{A flow chart illustrating Theorem \ref{thm:main}. Examples for reading this: ``\ref{as:CHT-PM} implies ACQ"; ``\ref{as:CHT-degeneracy} and \ref{as:PPHI-descent} are equivalent"; ``\ref{as:PPHI-cone} and ACQ jointly imply \ref{as:PPHI-ascent}."}
\label{fig:flow chart}
\end{figure}

\begin{proof}
Throughout this proof, consider a fixed $\theta^*$. We invoke Lemma \ref{lem:restate} and use the simplified versions of the assumptions.

\paragraph{1.}
If equalities are excluded, MFCQ reduces to
\begin{equation} \label{eq:MFCQ-like}
\min_{t \in \mathbb{R}^d: \Vert t \Vert =1} \max_{j \in \cJ^*(\theta^*)} D_j(\theta^*)t < 0.
\end{equation}
This immediately clarifies that excluding inequalities and imposing MFCQ is just Assumption \ref{as:PPHI-descent_2}. It remains to show equivalence with Assumption \ref{as:CHT-degeneracy}. 

Assumption \ref{as:CHT-degeneracy} excludes equalities because the presence of even one equality constraint implies that $\max_{j=1,\dots,J} E(m_j(X,\theta))/\sigma_j(\theta)=0$ for all $\theta \in \Theta_I$. This leaves two possibilities: Either $\Theta_I^{-\varepsilon} \neq \emptyset$, in which case \eqref{eq:CHT-deg-1} fails, or $\Theta_I^{-\varepsilon} = \emptyset$, in which case $\eqref{eq:CHT-deg-2}$ fails because $d(\theta,\Theta_I^{-\varepsilon})=\infty$.

To see that Assumption \ref{as:CHT-degeneracy} also implies \eqref{eq:MFCQ-like}, consider a sequence $\varepsilon_n \to 0$. For $n$ large enough we have $\varepsilon_n /M \leq \delta$, where $M$ and $\delta$ are from Assumption \ref{as:CHT-degeneracy}. Then by \eqref{eq:CHT-deg-2} there exists $\theta_n \in \Theta_I^{-\varepsilon_n/M}$ with $\Vert \theta_n-\theta^* \Vert \leq \varepsilon_n$, and by \eqref{eq:CHT-deg-1}, we have $\max_{j \in \cJ^*(\theta^*)} \{E(m_j(X,\theta_n))/\sigma_j(\theta_n)\} \leq -C \varepsilon_n/M$. Next, let $t$ be any accumulation point of $(\theta_n-\theta^*) / \Vert \theta_n-\theta^* \Vert$, then by continuous differentiability one has $\max_{j \in \cJ^*(\theta^*)} D_j(\theta^*)t \leq -C/M<0$.

To see the converse, let
\begin{equation}\label{eq:def_mu}
t^* = \arg\min_{t \in \mathbb{R}^d: \Vert t \Vert =1} \max_{j \in \cJ^*(\theta^*)} D_j(\theta^*)t
\end{equation}
and let $\mu=\tfrac{1}{2}\vert \max_{j \in \cJ^*(\theta^*)} D_j(\theta^*)t^* \vert$. By \eqref{eq:MFCQ-like}, $\max_{j \in \cJ^*(\theta^*)} D_j(\theta^*)t^*=-2\mu<0$.

We next argue why inactive constraints, i.e. $j \notin \cJ^*(\theta^*)$, can be ignored in what follows. Note first that by the Mean Value Theorem and Assumption \ref{as:momP_AS}, there exists $\tilde{\theta}$ componentwise between $\theta$ and $\theta^*$ s.t.
\begin{multline*}
\frac{E(m_j(X,\theta))}{\sigma_j(\theta)}-\frac{E(m_j(X,\theta^*))}{\sigma_j(\theta^*)} = D_j(\tilde{\theta}) ( \theta-\theta^* )  \\
= (D_j(\theta^*)+O(\Vert \theta-\theta^* \Vert)) ( \theta-\theta^* )  =  O(\Vert \theta-\theta^* \Vert) + O(\Vert \theta-\theta^* \Vert^2).
\end{multline*}
Let $\rho \equiv \max_{j \notin \cJ^*(\theta^*)}E(m_j(X,\theta^*))/\sigma_j(\theta^*)<0$, then it follows that
$$ \max_{\theta \in B(\theta^*,\eta)} \max_{j \notin \cJ^*(\theta)}E(m_j(X,\theta))/\sigma_j(\theta) \leq \rho+\max_{\theta \in B(\theta^*,\eta)} \left\{ \frac{E(m_j(X,\theta))}{\sigma_j(\theta)}-\frac{E(m_j(X,\theta^*))}{\sigma_j(\theta^*)}\right\}<0$$
for a small enough choice of $\eta$. Thus, inactive constraints do not affect the value taken by $\max_{j=1,\dots,J} E(m_j(X,\theta))/\sigma_j(\theta)$ anywhere on $B(\theta^*,\eta)$. (In the remainder of this proof, $\eta$ is understood to be either the $\eta$ just chosen or the $\eta$ from Assumption \ref{as:CHT-degeneracy}, whichever is smaller.)

Fix $\delta>0$ and consider any $\theta \in B(\theta^*,\eta)$ with $\max_{j \in \cJ^*(\theta^*)} E(m_j(X,\theta))/\sigma_j(\theta) > -\delta$. Then
\begin{eqnarray*}
&& \max_{j \in \cJ^*(\theta^*)} E(m_j(X,\theta-\delta t^*/\mu))/\sigma_j(\theta-\delta t^*/\mu) \\
& = & \max_{j \in \cJ^*(\theta^*)}\left\{  E(m_j(X,\theta))/\sigma_j(\theta) -\delta D_j(\bar{\theta}_j) t^*/\mu\right\} \\
& = & \max_{j \in \cJ^*(\theta^*)}\left\{  E(m_j(X,\theta))/\sigma_j(\theta) -\delta D_j(\theta^*) t^*/\mu - o(\delta)\right\} \\
&\geq & -\delta + 2\delta - o(\delta) > 0 
\end{eqnarray*}
for $\delta$ small enough, implying that $\theta-2\delta t^*/\mu \notin \Theta_I$ and therefore that $d(\theta,\Theta \setminus \Theta_I) < 2\delta/\mu$. (Here, $\bar{\theta}_j$ is componentwise between $\theta$ and $\theta-\delta t^*/\mu$ and may change with $j$; $\Theta_I \subset \operatorname{int}\Theta$ ensures $\theta-2\delta t^*/\mu \in \Theta$ for $\delta$ small enough.)

Conversely, by setting $\varepsilon=2\delta/\mu$, we find that (for $\varepsilon$ small enough)
\begin{equation*}
d(\theta,\Theta \setminus \Theta_I) \geq \varepsilon \implies \max_{j \in \cJ^*(\theta^*)} E(m_j(X,\theta))/\sigma_j(\theta) \leq -\mu\varepsilon/2,
\end{equation*}
verifying \eqref{eq:CHT-deg-1} with $C=\mu/2$. The requirement that $\varepsilon$ be small enough can be enforced by choosing $\eta$ low enough.

Next, for any $\theta \in \Theta_I \cap B(\theta^*,\eta)$, we similarly have
\begin{eqnarray*}
&& \max_{j \in \cJ^*(\theta^*)} E(m_j(X,\theta+\delta t^*/\mu))/\sigma_j(\theta+\delta t^*/\mu) \\
&=& \max_{j \in \cJ^*(\theta^*)} \left\{ E(m_j(X,\theta))/\sigma_j(\theta) + \delta D_j(\bar{\theta}_j)t^*/\mu  \right\} \\
&= & \max_{j \in \cJ^*(\theta^*)} \left\{ E(m_j(X,\theta))/\sigma_j(\theta) + \delta D_j(\theta^*)t^*/\mu + o(\delta) \right\} \\
& \leq & 0-2\delta + o(\delta)
\end{eqnarray*}
for $\delta$ small enough. (Again, $\theta+\delta t^*/\mu \in \Theta$ because $\Theta_I \subset \operatorname{int}\Theta$. The interpretation, though not the value taken, of $\bar{\theta}_j$ is as before.) 

Now, set $\overline{M}= \max_{j \in \cJ^*(\theta^*)} \Vert D_j(\theta^*) \Vert$ and let $t$ be any unit vector, then
\begin{eqnarray*}
&& \max_{j \in \cJ^*(\theta^*)} E(m_j(X,\theta+\delta t^*/\mu+\delta t/\overline{M}))/\sigma_j(\theta+\delta t^*/\mu+\delta t/\overline{M}) \\
&=& \max_{j \in \cJ^*(\theta^*)}\left\{ E(m_j(X,\theta+\delta t^*/\mu))/\sigma_j(\theta+\delta t^*/\mu) + \delta D_j(\bar{\theta}_j)t/\overline{M} \right\}  \\
&\leq & -2\delta +\delta + o(\delta) <0
\end{eqnarray*}
for $\delta$ small enough, thus 
\begin{equation*}
B(\theta+\delta t^*/\mu,\delta/\overline{M}) \subseteq \Theta_I \implies \theta+\delta t^*/\mu \in \Theta_I^{-\delta/\overline{M}} \implies d(\theta,\Theta_I^{-\delta/\overline{M}}) \leq \delta /\mu.
\end{equation*}
Setting $\varepsilon=\delta/\overline{M}$, we find $d(\theta,\Theta_I^{-\varepsilon}) \leq \varepsilon \overline{M}/\mu$, verifying \eqref{eq:CHT-deg-2} with $M=\overline{M}/\mu$.

\paragraph{2.}
Suppose that \eqref{eq:MFCQ-like} applies and let $\mu$ and $t^*$ be as in \eqref{eq:def_mu}. Fix any scalar
\begin{equation*}
\gamma \in \left(0,\frac{1}{\mu}\max_{\theta \in B(\theta^*,\eta)} \max_{j \in \cJ^*(\theta^*)} E(m_j(X,\theta))/\sigma_j(\theta) \right],
\end{equation*}
noting that by Assumption \ref{as:momP_AS}(d), the upper bound on $\gamma$ vanishes as $\eta \to 0$. Consider any $\theta \in B(\theta^*,\eta)$ s.t. $\max_{j \in \cJ^*(\theta^*)} E(m_j(X,\theta))/\sigma_j(\theta)<\mu\gamma$. Then by a use of the mean value theorem very similar to preceding displays, $$\max_{j \in \cJ^*(\theta^*)} E(m_j(X,\theta+\gamma t^*))/\sigma_j(\theta+\gamma t^*)<\mu\gamma-2\mu\gamma+o(\gamma)<0$$ for $\gamma$ small enough (which can be ensured by choosing $\eta$ small enough). It follows that $\theta+\gamma t^* \in \operatorname{int}\Theta_I$, hence $d(\theta,\Theta_I)<\gamma$. Conversely, $d(\theta,\Theta_I)=\gamma$ then implies $\max_{j \in \cJ^*(\theta^*)} E(m_j(X,\theta))/\sigma_j(\theta)\geq \mu\gamma$ and therefore $\max_j E(m_j(X,\theta))/\sigma_j(\theta)\geq \mu\gamma$. As $\gamma$ was arbitrary, this verifies Assumption \ref{as:CHT-PM_2} with $C=\mu$.  

\paragraph{3.}
Fix any $t \in \mathcal{L}(\theta^*)$. By differentiability of $E(m_j(X,\cdot))/\sigma_j(\cdot)$ and the definition of $\mathcal{L}(\cdot)$, we then have
\begin{equation*}
\lim \sup \bigl\{ n \times \max\bigl\{ \max_{j \in \mathcal{J}_1^*(\theta^*)} D_j(\theta^*)t/n, \max_{j \in \{J_1+1,...J\} } \vert D_j(\theta^*)t/n\vert \bigr\}\bigr\}  = 0 \text{ as }n \to \infty.
\end{equation*}
But then Assumption \ref{as:CHT-PM_2} implies $n \times d(\theta^*+t/n,\Theta_I) \to 0$. Next, let $\theta_n=\arg \min_{\theta \in \Theta_I}\Vert \theta - (\theta^*+t/n) \Vert$ and therefore $\Vert \theta_n-(\theta^*+t/n) \Vert=d(\theta^*+t/n,\Theta_I)$, then
\begin{multline*}
\lim \frac{\theta_n-\theta^*}{\Vert\theta_n-\theta^* \Vert}= \lim \frac{\theta_n-(\theta^*+t/n)+(\theta^*+t/n)-\theta^*}{\Vert \theta_n-(\theta^*+t/n)+(\theta^*+t/n)-\theta^*  \Vert} \\
=  \lim \frac{n \times (\theta_n-(\theta^*+t/n))+t}{\Vert n \times (\theta_n-(\theta^*+t/n))+t \Vert}=\frac{t}{\Vert t \Vert}, 
\end{multline*}
hence $t \in \mathcal{T}(\theta^*)$.

\paragraph{4.}
Under ACQ, $\cL(\theta^*)=\cT(\theta^*)$, hence Assumptions \ref{as:PPHI-ascent_2} and \ref{as:PPHI-cone_2} are then equivalent.

\paragraph{5.} This holds because $\cT(\theta^*)\subseteq \cL(\theta^*)$.

\paragraph{6.} Suppose the conclusion fails, thus no $d$ gradients of active constraints span $\mathbb{R}^d$. Then there exists a unit vector $t$ s.t. $D_j(\theta^*)t=0$ for all $j \in \mathcal{J}^*(\theta^*)$. This implies $$\min_{t \in \mathbb{R}^d:\Vert t \Vert=1,p't \geq 0} \max_{j \in \mathcal{J}^*(\theta^*)} D_j(\theta^*)t = 0$$ because at least one of $(t,-t)$ is feasible in this minimization problem, contradicting Assumption \ref{as:PPHI-ascent_2}; compare in particular the equivalent representation of this assumption on the right-hand side of \eqref{eq:conetransform}.

\paragraph{7.} Because $p't=0 \Leftrightarrow \theta^*+t \in H(p,\Theta_I)$, the right-hand side of \eqref{eq:conetransform}, and thereby Assumption \ref{as:PPHI-ascent_2}, can be restated as $$\min_{\theta \in H(p,\Theta_I)} \max\bigl\{ \max_{j \in \mathcal{J}_1^*(\theta^*)} D_j(\theta^*)(\theta-\theta^*)/\Vert\theta-\theta^*\Vert, \max_{j \in \{J_1+1,...J\} } \vert D_j(\theta^*)(\theta-\theta^*)\vert/\Vert\theta-\theta^*\Vert \bigr\}  > 0.$$ By continuous differentiability (using arguments very similar to above), this then implies that, for small enough $\eta>0$,
\begin{multline*}
\max\bigl\{0,\max_{j=1,\dots,J_1}E(m_j(X,\theta))/\sigma_j(\theta),\max_{j=J_1+1,\dots,J}\vert E(m_j(X,\theta))/\sigma_j(\theta)\vert\bigr\} \geq C \Vert \theta-\theta^* \Vert \\
\text{for all} ~ \theta \in  B(\theta^*,\eta) \cap H(p,\Theta_I),
\end{multline*}
where $C = \tfrac{1}{2} \min_{t \in \mathbb{R}^d:\Vert t \Vert=1,p't = 0} \max\bigl\{ \max_{j \in \mathcal{J}_1^*(\theta^*)} D_j(\theta^*)t, \max_{j \in \{J_1+1,...J\} }  D_j(\theta^*)t \bigr\}>0.$
\end{proof}

\subsection{Tightness of Theorem \ref{thm:main}} \label{app:examples}

We next clarify that Theorem \ref{thm:main} is tight, i.e., no implication that is not indicated in Figure \ref{fig:flow chart} holds true. This subsection can be skipped without loss of continuity.

Let $m_j(X,\theta) = \mu_j(\theta) +X_j$ for functions $\mu_j:\Theta \mapsto \mathbb{R}$ defined below and suppose that all $X_j$ are standard normal; thus, $E(m_j(X,\theta))/\sigma_j(\theta)=\mu_j(\theta)$.  Also, $\Theta=[-1,1]^2$. All examples are constructed s.t. the support point of interest is $\theta^*=(0,0)$. Direction of projection is $p=(0,1)$ unless explicitly indicated otherwise. There are no equality constraints, so for the purpose of these examples, ``MFCQ and no equalities" is just MFCQ.

\paragraph{Neither Assumption \ref{as:PPHI-cone_2} nor \ref{as:BCS-PM_2} imply either \ref{as:PPHI-ascent_2} or \ref{ACQ} (ACQ).}
\begin{eqnarray*}
\mu_1(\theta) &=& \theta_2^3 - \theta_1 \\
\mu_2(\theta) &=& \theta_2^3 + \theta_1.
\end{eqnarray*}
We start with this example because several others build on it. Qualitatively, it is the left panel of Figure \ref{fig:support point}. The linear cone $\cL(\theta^*)$ is spanned by $\{(0,1),(0,-1)\}$, the tangent cone $\cT(\theta^*)$ is spanned only by $\{(0,-1)\}$. The one direction, $t=(0,-1)$,  that is in $\cT(\theta^*)$ is tangential to all constraints. The example violates ACQ (and all stronger assumptions) as well as \ref{as:PPHI-ascent_2} but fulfills \ref{as:PPHI-cone_2} as well as \ref{as:BCS-PM_2}.  

\paragraph{Assumption \ref{as:CHT-PM_2} does not imply \ref{MFCQ} (MFCQ).}
\begin{eqnarray*}
\mu_1(\theta) &=& \theta_2^3 - \theta_1 \\
\mu_2(\theta) &=& \theta_2^3 + \theta_1 \\
\mu_3(\theta)&=&\theta_2.
\end{eqnarray*}
This example adds a third constraint to the first example; compare the right panel of Figure \ref{fig:support point}. Both $\cL(\theta^*)$ and $\cT(\theta^*)$ are spanned by $(0,-1)$. The example therefore fulfills \ref{as:CHT-PM_2} (and by implication ACQ) but not MFCQ.

\paragraph{Assumption \ref{ACQ} (ACQ) does not imply \ref{as:CHT-PM_2}.}
\begin{eqnarray*}
\mu_1(\theta) &=& \theta_1^2(\theta_2+\theta_1^2) \\
\mu_2(\theta) &=& \theta_2.
\end{eqnarray*}
Here, $\cL(\theta^*)=\cT(\theta^*)=\{t:p't\leq 0\}$, so that ACQ is fulfilled. However, \ref{as:CHT-PM_2} is violated in direction $t=(1,0)$:
\begin{eqnarray*}
\mu_1(\gamma t) &=& \gamma^4 \\
d(\gamma t,\Theta_I) &\leq & \gamma^2~~~~~\text{(because}~~(\gamma,-\gamma^2)\in \Theta_I\text{)} \\
\implies \mu_1(\gamma t)/d(\gamma t,\Theta_I) & \to & 0~~~~~\text{as}~~\gamma\to 0.
\end{eqnarray*}

\paragraph{Assumption \ref{as:PPHI-ascent_2} does not imply \ref{ACQ} (ACQ).}
\begin{eqnarray*}
\mu_1(\theta) &=& -\theta_1+\theta_2 \\
\mu_2(\theta) &=& \theta_1+\theta_2 \\
\mu_3(\theta) &=& (\theta_1 \theta_2)^2.
\end{eqnarray*}
In this example (which is inspired by a well-known counterexample to ACQ), we have that $\cL(\theta^*)$ is spanned by $\{(-1,-1),(1,-1)\}$ but $\cT(\theta^*)$ is spanned by $(0,-1)$ only.

\paragraph{Assumption \ref{as:BCS-PM_2} does not imply \ref{as:PPHI-ascent_2} or \ref{as:PPHI-cone_2}.}
\begin{eqnarray*}
\mu_1(\theta) &=& \theta_2 \\
\mu_2(\theta) &=& \theta_1+\theta_2.
\end{eqnarray*}
In this example, $\cL(\theta^*)$ and $\cT(\theta^*)$ coincide and are spanned by $\{(-1,0),(1,-1)\}$, contradicting both \ref{as:PPHI-ascent_2} and \ref{as:PPHI-cone_2}. Assumption \ref{as:BCS-PM_2}, which here only applies if we move in direction $(1,0)$ from $\theta^*$, is fulfilled.

\subsection{Discussion}

Our findings inform a number of clarifying remarks on the existing literature.
\begin{itemize}
\item As mentioned above, CHT's polynomial minorant can be disentangled into a local and a global identification condition. The local condition \eqref{eq:CHT-orig-local} is a mild strengthening of ACQ and is implied by degeneracy. The global condition \eqref{eq:CHT-orig-global} is essentially the weakest additional statement needed to ensure that $\Theta_I$ is a well-separated (if set-valued) minimum of $\max\bigl\{\max_{j=1,\dots,J_1} E(m_j(X,\theta))/\sigma_j(\theta), \max_{j=J_1+1,\dots,J} |E(m_j(X,\theta))/\sigma_j(\theta)|,0\bigr\}$.\footnote{A well-separated minimum requires that for each $\varepsilon>0$, there exists $\delta>0$ s.t. $d(\theta,\Theta_I) \geq \varepsilon$ implies $\max\{\dots\} \geq \delta$. Its role in ensuring ``inner consistency" of sample analogs of identified sets is well understood \citep{Newey_McFadden1994a}.} While the polynomial minorant condition is, therefore, not redundant, an instructive restatement of the assumptions is available.   
\item Regarding assumptions in PPHI, claim 5 of Theorem \ref{thm:main} owes to our simplification, but in view of the smoothness imposed in their Assumption 7, claim 4 also applies to the original versions.

Also, if $J=d$, then the PPHI assumptions imply LICQ. This clarifies relation to recent work by \cite{ChoRussell} and \cite{Gafarov}: Both effectively impose LICQ and benefit from this by being able to propose relatively simple inference. However, while stronger than assumptions in CHT, BCS, and certainly \cite{KMS}, the assumptions exceed those in PPHI only in the sense of excluding ``overidentified" support points, i.e. more than $d$ active constraints, and are actually weaker in other respects.
\item \cite{ChoRussell} furthermore impose Assumption \ref{as:CHT-PM}, i.e. (by Theorem \ref{thm:main}) ACQ. This is not redundant in their paper because \ref{as:CHT-PM} is imposed for all $\theta \in \Theta_I$ and with universal $C$, i.e. ``more uniformly" than LICQ.
\item \cite{Yildiz2012} presents conditions for Hausdorff consistency of $\hat{\Theta}_I$. Some of these are in essence constraint qualifications and can be related to our analysis as follows. For the case of pure inequality constraints, her high-level Assumption 3.1 states that $\Theta_I$ is the closure of the \textit{strict} level $0$ lower contour set of $\max_j E(m_j(X,\theta))/\sigma_j(\theta)$.\footnote{\cite{Molchanov1998} also uses this condition to ensure Hausdorff consistency of set estimators.} For the case of at most $d$ active constraints, this is derived as an implication of the LICQ (Assumption 3.2). For the more general case, it is derived from an assumption (in Lemma 3.1) enforcing that, at any boundary point $\theta^*$ of $\Theta_I$, the criterion function $\max_j E(m_j(X,\theta))/\sigma_j(\theta)$ is strictly increasing in some component of $\theta$. To make it comparable to our assumptions, one would impose this to hold at any support point $\theta^*$. By a minimal extension of step 2 of Theorem \ref{thm:main} (see expression \eqref{eq:MFCQ-like}), it is then equivalent to Assumption \ref{as:CHT-degeneracy}. Therefore, \cite{Yildiz2012} essentially imposes MFCQ for the pure inequality case.\footnote{\cite{Yildiz2012} writes that her assumptions imply a degeneracy condition in CHT without claiming the reverse implication. This refers to their high-level degeneracy assumption C.3, which is implied whenever the simple sample analog of $\Theta_I$ is consistent.} For the case of mixed equalities and inequalities, she invokes a LICQ (Assumption 4.1(b)).
\item We close with some remarks on why, for inference on projections $p'\theta$, certain approaches do not require constraint qualifications. Specifically, the profiling approach in BCS gets by with the relatively weak Assumption  \ref{as:BCS-PM}; \cite{KMS} or projection of confidence regions for $\theta$ (such as those in \cite{AndrewsSoares2010E}, \cite{Bugni2009E}, or \cite{Canay2010JE}) use no shape restrictions for $\Theta_I$ at all.

The reason is that all of these approaches localize inference at a conjectured true value of the support function $s(p,\Theta_I)$ (in BCS) or parameter vector $\theta$ (in all others). Consistent estimation of identified sets for these objects is then not a concern. In particular, BCS need to ensure some form of consistency of a sample analog of $\Theta_I$ that is restricted to the true supporting hyperplane, and this is precisely what the minorant on the support plane achieves. The other approaches need no constraint qualifications at all.

Of course, there is no free lunch. All the methods just alluded to are computationally expensive because they effectively invert tests whose critical value depends on the exact value of the parameter under the null. Thus, while some shortcuts may be available in practice (see in particular \cite{KMS} and \cite{KMST_code}), critical values must in principle be recomputed at every conceivable value of $\theta$ or $s(p,\Theta)$.
\end{itemize}
\newpage

\section{Verifying Assumptions in Some Examples}
\label{sec:examples_verify_assumptions}

In this section, we discuss how these restrictions apply to two well-understood examples of partial identification. It will become clear that one cannot take for granted that they ``typically" hold and that verifying them could be intricate in more involved examples. That is part of our message: We point out that many of them are basically constraint qualifications, and it is well known that constraint qualifications can be subtle to verify.

The examples are visualized in Figure \ref{fig:examples2}, which is designed to resemble Figure \ref{fig:cones}. Note that in the second example, $\Theta_I$ has zero measure.

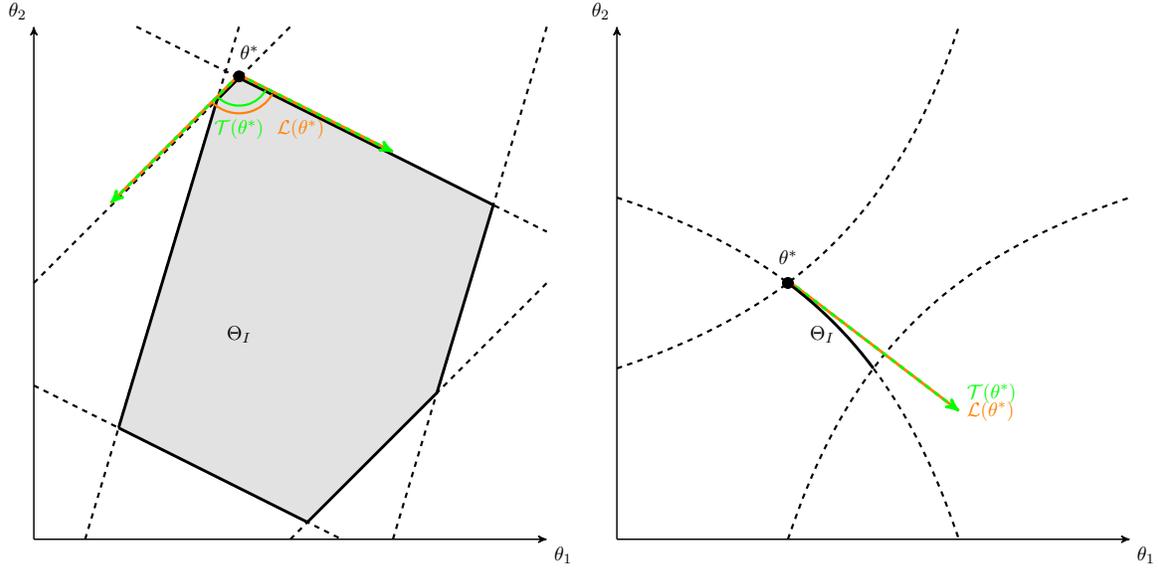
\begin{figure}[t]
\begin{subfigure}[b]{.5\linewidth}
\begin{adjustbox}{max width=\textwidth}
\tikzstyle{line} = [draw, -latex']
\begin{tikzpicture}[node distance = 2cm, auto]
\linespread{0.9}

\draw [->, thick] (0,0) to (10,0) node[below right] {$\theta_1$};
\draw [->, thick] (0,0) to (0,10) node[above left] {$\theta_2$};
\draw[dashed,very thick] (1,0) to (4,10);
\draw[dashed,very thick] (7,0) to (10,10);
\draw[dashed,very thick] (0,3) to (6,0);
\draw[dashed,very thick] (2,10) to (10,6);
\draw[dashed,very thick] (0,5) to (5,10);
\draw[dashed,very thick] (5,0) to (10,5);
\path [fill=darkgray!15] (38/23,50/23) to (25/7,60/7) to (4,9) to (206/23,150/23) to (55/7,20/7) to (16/3,1/3) to (38/23,50/23);
\draw[black,ultra thick] (38/23,50/23) to (25/7,60/7) to (4,9) to (206/23,150/23) to (55/7,20/7) to (16/3,1/3) to (38/23,50/23);
\draw[-,ultra thick,orange] (4,9.06) to (7,7.56);
\draw[->,dashed,ultra thick,green] (4,9.06) to (7,7.56);
\draw[-,ultra thick,orange] (4,9.06) to (1.5,6.56);
\draw[->,dashed,ultra thick,green] (4,9.06) to (1.5,6.56);
\centerarc[green, very thick](4,9.06)(225:330:.6);
\centerarc[orange, very thick](4,9.06)(225:330:.75);
\filldraw [black] (4,9.03) circle (3pt);
\node (theta^*) at (4.2,9.5) {$\theta^*$};
\node (theta_I) at (4,4) {$\Theta_I$};
\node[green] (tangent) at (4,8) {$\cT(\theta^*)$};
\node[orange] (linear) at (5.2,8) {$\cL(\theta^*)$};
\end{tikzpicture}
\end{adjustbox}
               \end{subfigure}%
               \begin{subfigure}[b]{.5\linewidth}
               \begin{adjustbox}{max width=\textwidth}
\tikzstyle{line} = [draw, -latex']
\begin{tikzpicture}[node distance = 2cm, auto]
\linespread{0.9}

\draw [->, thick] (0,0) to (10,0) node[below right] {$\theta_1$};
\draw [->, thick] (0,0) to (0,10) node[above left] {$\theta_2$};
\draw[domain=0:20/3,smooth,variable=\x,dashed,very thick] plot ({\x},{10-10/(3-.3*\x)});
\draw[domain=0:20/3,smooth,variable=\x,dashed,very thick] plot ({\x},{10/(3-.3*\x)});
\draw[domain=10/3:10,smooth,variable=\x,dashed,very thick] plot ({\x},{10-10/(.3*\x)});
\draw[-,ultra thick,orange] (10/3,5.05) to (20/3,2.505) node[right] {$\cL(\theta^*)$};
\draw[->,dashed,ultra thick,green] (10/3,5.05) to (20/3,2.505) node[above right] {$\cT(\theta^*)$};
\draw[domain=10/3:5,smooth,variable=\x,ultra thick] plot ({\x},{10-10/(3-.3*\x)}) ;
\filldraw [black] (10/3,5) circle (3pt);
\node (theta^*) at (10/3,5.5) {$\theta^*$};
\node (theta_I) at (4,4) {$\Theta_I$};
\end{tikzpicture}
\end{adjustbox}
\end{subfigure}
\caption{Visualization of the two examples from Section \ref{sec:examples_verify_assumptions}. Dashed lines denote (in)equality restrictions. The figures also display the identified set $\Theta_I$, the support point $\theta^*$ in direction $p=(0,1)$ (i.e. ``up," as in previous figures), and the associated linear and tangent cones (which coincide). Left panel: Linear regression with interval outcome data (Section \ref{sec:MD}). Right panel: A simple entry game (Section \ref{sec:EG}).}
\label{fig:examples2}
\end{figure}

\subsection{Linear regression with interval outcome data and discrete regressors}
\label{sec:MD}	
Consider a linear regression model:
\begin{align*}
	W=Z^\prime \theta +u,
\end{align*}
where $Z=(Z_1,\dots,Z_d)$ is a $d\times 1$ random vector with $Z_1=1$. We assume that $Z$ has $k$ points of support denoted $\mathcal Z=\{z^1,\dots,z^k \in \R^d\}$ with $\max_{r=1,\dots,k}\Vert z^r \Vert<M<\infty$. The researcher observes $\{W_0,W_1,Z\}$ with $P(W_0\le W \le W_1|Z=z^r)=1,r=1,\dots,k$, where $\pi^r=P(Z=z^r)>0,r=1,\dots,k$ are assumed known. Suppose that $W_0$ and $W_1$ take values in $\mathcal W\subset\mathbb R$. An important special case is missing data, where $W_0$ and $W_1$ both equal the true $W$ if the latter is observed and correspond to some bound on it otherwise.

The identified set is characterized by the following moment inequalities.
\begin{align*}
E(W_0|Z=z^r)-z^{r\prime} \theta &\le 0,~r=1,\dots,k,\\
z^{r\prime} \theta-E(W_1|Z=z^r)& \le 0,~ r=1,\dots,k.
\end{align*}
Equivalently,
\begin{align}
E(W_01\{Z=z^r\})/\pi^r-z^{r\prime} \theta& \le 0,~ r=1,\dots,k,\label{eq:linear1}\\
z^{r\prime} \theta-E(W_11\{Z=z^r\})/\pi^r& \le 0,~ r=1,\dots,k.\label{eq:linear2}
\end{align}
Define the following objects.
\begin{align*}
\sigma_j&=\begin{cases}
		\operatorname{Var}(W_01\{Z=z^j\})^{1/2}/\pi^j,&~j=1,\dots,k\\
		\operatorname{Var}(W_11\{Z=z^{j-k}\})^{1/2}/\pi^{j-k},&~j=k+1,\dots,2k,
	\end{cases}\\
	D_j&=\begin{cases}
		-z^{j\prime}/\sigma_j,&~j=1,\dots,k\\
		z^{j-k\prime}/\sigma_j,&~j=k+1,\dots,2k,
	\end{cases}\\
	b_j&=\begin{cases}
	-\frac{E(W_0|Z=z^j)}{\operatorname{Var}(W_0|Z=z^j)^{1/2}}	&~j=1,\dots,k\\
	\frac{E(W_1|Z=z^{j-k})}{\operatorname{Var}(W_1|Z=z^{j-k})^{1/2}}	&~j=k+1,\dots,2k.
	\end{cases}
\end{align*}
Note that, since $D_j$ does not depend on $\theta$, we drop its argument.

By \eqref{eq:linear1}-\eqref{eq:linear2}, the identified set is a polytope characterized by $k$ pairs of parallel constraints
	\begin{align}
	E(W_01\{Z=z^r\})/\pi^r\le	z^{r\prime}\theta \le \theta-E_P(W_11\{Z=z^r\})/\pi^r,~r=1,\dots,k. \label{eq:response}
	\end{align}
See the left panel of Figure \ref{fig:examples2} for an illustration with $k=3$. Each support point is either a vertex of the polytope or a point in the relative interior of a non-singleton support set. We therefore analyze subcases below. Note that, since the gradients $z^1,\dots,z^k$ of the constraints in \eqref{eq:response} are known, one knows without data which case applies.

We first establish conditions implying MFCQ. For $j=1,\dots,2k$, define:
\begin{align*}
	H_j&=\{\theta\in\Theta:D_j\theta=b_j\}\\
	H_j^-&=\{\theta\in\Theta:D_j\theta\le b_j\}.
\end{align*}
We call $H_j$ a \emph{hyperplane} and $H_j^-$ a \emph{half space}. Let $\operatorname{ri}(A)$ denote the relative interior of $A$. A $(d-1)$-dimensional flat face in $\mathbb R^d$ is called a \emph{facet}. Similarly, a $\ell$-dimensional element of a $d$-dimensional polytope is called a $\ell$-face, where $1\le\ell\le d-2$. For example, a $1$-face is an edge in a 3 dimensional polytope.
\begin{lemma}\label{lem:mfcq1}
	Suppose (i) $\mathcal W$ is compact; (ii) $\Theta=\{\theta\in\mathbb R^d:\|\theta\|^2\le B_0\}$ with $B_0<\infty$ satisfying $C_0B_0>k\sup_{w\in\mathcal W}w^2$, where $C_0\equiv \inf_{p\in\mathbb S^{d-1}}\sum_r (p'z^r)^2$; (iii) $k\ge d$, and any subset $\mathcal C\subseteq \mathcal Z$ with $\#\mathcal C \le d$ is linearly independent; (iv) $E[W_1-W_0|Z=z^r]>0$ for all $r=1,\dots,k.$ Let $\theta^*\in S(p,\Theta_I)$. 
	\begin{enumerate}
		\item (Facet) If $S(p,\Theta_I)$ is not a singleton and is the intersection of a hyperplane and a finite number of half spaces,  MFCQ holds at any $\theta^*\in \text{ri}(S(p,\Theta_I))$;
		\item ($\ell$-face) If $S(p,\Theta_I)$  is not a singleton and is  the intersection of more than one hyperplanes, MFCQ holds at any  $\theta^*\in \text{ri}(S(p,\Theta_I))$;
		\item (Vertex) If $S(p,\Theta_I)$ is a singleton and $\#\mathcal J^*(\theta^*)\le d$, MFCQ holds at $\theta^*$.
	\end{enumerate} 
\end{lemma}

\begin{remark}
	The lemma shows that, even if one is not sure about whether $\theta^*$ is on a facet or any other lower dimensional elements of the polytope, MFCQ holds as long as $\#\mathcal J^*(\theta^*)\le d$ (and other conditions of the lemma hold).
	Also, things simplify if $k=d$, in which case conditions (iii) and (iv) imply $\#{\mathcal J}^*(\theta^*)\le d$ at any support point. Condition (iii) then ensures LICQ, and therefore MFCQ, at $\theta^*$.
\end{remark}

\begin{proof}
Under our assumptions, $\Theta_I$ is nonempty and is in the interior of $\Theta$ by Proposition F.1 in \cite{Kaido:Santos}.

\medskip
\noindent
\textbf{Case 1:} Suppose $S(p,\Theta_I)$ is a flat face and $\theta^*\in \operatorname{ri}(S(p,\Theta_I))$. This occurs if and only if $p=cz^r$ (or $p=-cz^r$) for some $r\in\{1,\dots, k\}$ and $c>0$, and $p'\theta^*=E(W_11\{Z=z^r\})/\pi^r$ (or $E(W_01\{Z=z^r\})/\pi^r$). By (iii), such $r$ is unique.
 Without loss of generality, suppose $p'\theta^*=E(W_11\{Z=z^r\})/\pi^r$. The assumption $E[W_1-W_0|Z=z^r]>0$ ensures
\begin{multline}
E(W_11\{Z=z^r\})/\pi^r-E(W_01\{Z=z^r\})/\pi^r  \\
=E[E[W_1-W_0|Z=z^r]1\{Z=z^r\}]/\pi^r>0.\label{eq:mfcq1_1}
\end{multline}
Hence, $p'\theta^*>E(W_0 1\{Z=z^r\})/\pi^r$, implying the lower bound is slack.
This and $\theta^*\in \operatorname{ri}(S(p,\Theta_I))$ imply $\mathcal J^*(\theta^*)=\{j^*\}$ with $j^*=r+k$. 

The gradient of the normalized moment is
\begin{align*}
	D_{j^*}=\frac{z^{r\prime}}{\sigma_{j^*}}=\frac{z^{r\prime}}{\operatorname{Var}(W_11\{Z=z^r\})^{1/2}/\pi^r}.
\end{align*}
Let $t^*=-z^r$, then 
\begin{align*}
	D_{j^*}t^*=-\frac{z^{r\prime}z^r}{\operatorname{Var}(W_11\{Z=z^r\})^{1/2}/\pi^r}<0,
\end{align*}
which establishes MFCQ.

\medskip
\noindent
\textbf{Case 2:} We first claim that $\#\mathcal J^*(\theta^*)<d$. Suppose otherwise. Then, there are at least $d$ active inequalities at $\theta^*$. Select a subset $\tilde{\mathcal J}\subset \mathcal J(\theta^*)$ such that $\#\tilde{\mathcal J}=d$. Then, by condition (iii), $\{D_j,j\in\tilde J\}$ are linearly independent, implying $\theta^*$ is a unique solution to the system of  linear equations 
\begin{align*}
	D_j\theta=b_j,~\forall j\in\tilde{\mathcal J}.
\end{align*}
Furthermore, by the necessary condition of the maximization problem, there is $\lambda\in\mathbb R^{2k}_+$ such that
\begin{align}
	p+\sum_{j\in\tilde{\mathcal J}}\lambda_j D_j'=0,\label{eq:foc1}
\end{align}
where $\lambda_j>0$ for all $j\in\tilde{\mathcal J}$.

Let $\tilde\theta\in S(p,\Theta_I)$ and $\tilde\theta\ne\theta^*$. By construction, $p'(\tilde\theta-\theta^*)=0$.
This and \eqref{eq:foc1} imply
\begin{align}
	\sum_{j\in\tilde{\mathcal J}}\lambda_j D_j(\tilde\theta-\theta^*)=0.\label{eq:foc2}
\end{align}
Suppose that $D_j(\tilde\theta-\theta^*)>0$ for some $j\in\tilde{\mathcal J}$. This implies
\begin{align*}
	D_j\tilde\theta>b_j,
\end{align*}
hence the $j$-th constraint is violated, hence $\tilde\theta$ cannot be in $S(p,\Theta_I)$, a contradiction. Similarly, suppose that $D_j(\tilde\theta-\theta^*)<0$ for some $j\in\tilde{\mathcal J}$, then by \eqref{eq:foc2} and the positivity of the Lagrange multipliers, there must exist $j'\in\tilde{\mathcal J}$ such that $D_{j'}(\tilde\theta-\theta^*)>0$, and the same argument applies. The only remaining possibility is $D_j(\tilde\theta-\theta^*)=0$ for all $j\in\tilde{\mathcal J}$, in which case
\begin{align*}
	D_j\tilde\theta=b_j,~\forall j\in\tilde{\mathcal J},
\end{align*}
which contradicts  $\theta^*$ being the unique solution to \eqref{eq:foc1}. Therefore, $\#\mathcal J^*(\theta^*)<d$ must hold.
Now, by condition (iii) and $\#\mathcal J^*(\theta^*)<d$,  LICQ holds at $\theta^*$, which implies the claim.

\medskip
\noindent
\textbf{Case 3:} By $\#\mathcal J^*(\theta^*)\le d$ and condition (iii), LICQ holds at $\theta^*$, which implies the claim.

\end{proof}

We conclude by mentioning other conditions. The identified set $\Theta_I$ is a finite polytope, hence any projection is the value of a linear program. This immediately clarifies that ACQ obtains, which furthermore means that Assumptions \ref{as:PPHI-cone} and \ref{as:PPHI-ascent} coincide. Both restrict $\cL(\theta^*)$, which locally just coincides with $\Theta_I$, to not intersect the supporting hyperplane other than at $\theta^*$. This will hold iff the support set is a singleton, i.e. a vertex. Finally, the solution to a linear program is necessarily well-separated, so that Assumption \ref{as:BCS-PM} holds.

\subsection{A Simple Entry Game}
\label{sec:EG}

Two-player entry games are a workhorse example of partial identification since \cite{Tamer03}. We analyze the game specified in BCS but without covariates. The equilibrium concept is pure strategy Nash equilibrium (PSNE). Firm $1$ respectively $2$ enters if
\begin{eqnarray*}
\varepsilon_1 - \theta_1 A_2  &\geq &  0 \\
\varepsilon_2-\theta_2 A_1  &\geq &  0,
\end{eqnarray*} 
where $(A_1,A_2)\in\{0,1\}^2$ are the firms' actions, $(\theta_1,\theta_2)\in [0,1]^2$ are the interaction parameters, and $(\varepsilon_1,\varepsilon_2)$ are observed by the players but not by the researcher. This system is incomplete: For certain realization of $(\varepsilon_1,\varepsilon_2)$, both $(1,0)$ and $(0,1)$ are PSNE and the model is silent on which is played. Hence, different (stochastic) selection mechanisms picking an equilibrium in the region of multiplicity can be coupled with different values of $\theta$ to yield the same distribution of observables $(A_1,A_2)$.
Nonetheless, inference can be carried out by bounding the likelihood of different outcomes from below by the respective probabilities of them being unique PSNE. In this simple example, this exhausts all the information in the model and data, leading to a sharp identification region. See \cite{BMM_ECMA} for characterizations of sharp identification regions in more complex models.

In this model,
\begin{itemize}
\item $(1,1)$ is the unique PSNE if $\varepsilon_1-\theta_1 A_2  \geq 0$ and $\varepsilon_2-\theta_2 A_1 \geq 0$,
\item $(1,0)$ is the unique PSNE if $\varepsilon_1-\theta_1 A_2  \geq 0$ and $\varepsilon_2-\theta_2 A_1 \leq  0$,
\item $(0,1)$ is the unique PSNE if $\varepsilon_1-\theta_1 A_2  \leq  0$ and $\varepsilon_2-\theta_2 A_1  \geq  0$.
\end{itemize}
Letting $\pi_{jk}=\Pr(A_1=j,A_2=k)$ and assuming (as in BCS) that $(\varepsilon_1,\varepsilon_2)$ are distributed independently uniformly on $[0,1]$, we have
\begin{eqnarray}
\pi_{11} &=& (1-\theta_1)(1-\theta_2)  \label{eq:cond1} \\
\pi_{10} &\geq & (1-\theta_1)\theta_2 \label{eq:cond2} \\
\pi_{01} &\geq & \theta_1(1-\theta_2) \label{eq:cond3} \\
\pi_{00} &=& 0.
\end{eqnarray}
Geometrically, the identified set $\Theta_I$ is an arc segment. The right panel of Figure \ref{fig:examples2} depicts its true shape if $\pi_{01}=\pi_{10}=\pi_{11}=1/3$.

Consider now the problem of bounding $\theta_2$ from above; all other bounds on individual parameters are similar. Guessing that \eqref{eq:cond1}-\eqref{eq:cond2} will bind at the solution, one can easily solve for
\begin{equation*}
\theta^*=\left[\begin{array}{c} \theta_1^* \\ \theta_2^* \end{array}\right]
=\left[\begin{array}{c} \pi_{01} \\ \pi_{10}/(\pi_{10}+\pi_{11}) \end{array}\right].
\end{equation*}
Furthermore, the linear and tangent cone are tangent to \eqref{eq:cond1} and are spanned by $(\pi_{10}+\pi_{11},-\pi_{11}/(\pi_{10}+\pi_{11}))$. This vector has strictly negative inner product with $p=(0,1)$. We conclude:
\begin{itemize}
\item Assumption \ref{as:CHT-degeneracy} and any equivalent assumption cannot hold because $\Theta_I$ has no interior.
\item All of Assumptions \ref{as:CHT-PM}, \ref{as:BCS-PM}, \ref{as:PPHI-cone}, \ref{as:PPHI-ascent}, LICQ, and ACQ hold.
\end{itemize}
However, fulfillment of some of these assumptions is delicate in that it depends on $\theta^*$ being an exposed point of $\Theta_I$. Other directions of optimization (e.g., $p=(1,-1)$, corresponding to testing the null that $\theta_1=\theta_2$) have $\theta^*$ in the relative interior of $\Theta_I$, i.e. at a point where only \eqref{eq:cond1} is active. Assumptions \ref{as:BCS-PM}, \ref{as:PPHI-cone}, and \ref{as:PPHI-ascent} will then fail.

\section{Conclusion}
\label{sec:conclusion}

The literature on partial identification uses constraint qualifications in many ways: To ensure Hausdorff consistency or rates of convergence for simple estimators of identified sets \citep[CHT; ][]{Yildiz2012}, to justify inference for the full parameter vector $\theta$ (CHT) or subvectors \citep[PPHI; ][]{ChoRussell,Gafarov}, or to justify efficiency bounds \citep{Kaido:Santos}. However, some of these uses are implicit, making it difficult even for expert readers to compare assumptions. We provide a guide to how different high-level assumptions relate to each other and to well-known constraint qualifications. A simple, important message is that several high-level assumptions are tightly related to the Mangasarian-Fromowitz constraint qualification and are essentially mutually equivalent. We believe that this provides useful guidance to readers trying to make sense of the large menu of inference methods for partially identified vectors and subvectors \citep{CS17,MolinariHOE}. For example, it clarifies costs and benefits relative to work that has weak-to-no geometric regularization, mostly at the expense of computational effort.

%\appendix
%\appendixpage
%

\bibliography{KMS,projection}

\end{document}